\documentclass{svjour3}                     % onecolumn (standard format)
\smartqed  % flush right qed marks, e.g. at end of proof
\usepackage{color}
\usepackage{amsfonts,color,pslatex}
\usepackage{amssymb,amsmath,latexsym}

\usepackage[caption=true]{caption}
\usepackage[para]{threeparttable}

\newcommand{\tr}{{\rm Tr}}
\newcommand{\gf}{{\mathbb{F}}}

\newcommand{\C}{{\mathcal C}}

\newcommand{\bc}{{\mathbf{c}}}

\usepackage{graphicx}

\begin{document}

\title{Linear Codes with Two or Three Weights From Quadratic Bent Functions}
%\thanks{}



\author{Zhengchun Zhou  \and
       Nian Li   \and   Cuiling Fan \and Tor Helleseth}


\institute{Z. Zhou  \at School of Mathematics, Southwest Jiaotong University, Chengdu, Sichuan, 610031, China.\\
\email{zzc@home.swjtu.edu.cn}      \\      %  \\
\and N. Li  \at Department of Informatics, University of Bergen, N-5020 Bergen, Norway. \\
\email{nian.li@ii.uib.no}    \\
\and C. Fan \at School of Mathematics, Southwest Jiaotong University, Chengdu, Sichuan, 610031, China.\\
\email{fcl@home.swjtu.edu.cn} \\
\and T. Helleseth \at Department of Informatics, University of Bergen, N-5020 Bergen, Norway.\\
\email{tor.helleseth@ii.uib.no} }

\date{Received: date / Accepted: date}

\maketitle

\begin{abstract}
Linear codes with few weights have applications in secrete sharing, authentication codes, association
schemes, and strongly regular graphs.  In this paper, several classes of $p$-ary
linear codes with two or three weights are constructed from quadratic Bent functions over the finite field $\gf_p$, where $p$ is an odd prime. They include some earlier linear codes
 as special cases.  The weight distributions
of these linear codes are also determined.

\keywords{Linear code, optimal code, Bent function,  quadratic form, weight distribution.}
% \PACS{PACS code1 \and PACS code2}
\subclass{ 94A24 \and 94B35 \and 94B15 \and 94A55}
\end{abstract}

\section{Introduction}\label{sec-int}
Throughout this paper, let $p$ be an odd prime and $m$ be a positive integer.
An $[n,\kappa,d]$ linear code over the finite field $\gf_p$ is a $\kappa$-dimensional subspace of $\gf_p^n$
with minimum (Hamming) distance $d$. Let $A_i$ denote the number of codewords
with Hamming weight $i$ in a code $\mathcal{C}$ of length $n$. The weight enumerator of $\mathcal{C}$
is defined by
\begin{eqnarray*}
1+A_1z+A_2z^2+\cdots+A_nz^n.
\end{eqnarray*}
The sequence $(A_1,A_2,\cdots,A_{n})$ is called the weight distribution of the code. Clearly, the weight
distribution gives the minimum distance of the code, and thus the error correcting capability. In addition, the
weight distribution of a code allows the computation of the error probability of error detection and correction
with respect to some error detection and error correction algorithms (see \cite{Klov} for details).
Thus the study of the weight distribution of a linear code is an important research topic in coding theory.
A linear code $\C$ is said to be $t$-weight if the number of nonzero $A_i$ in the sequence $(A_1,A_2,\cdots,A_{n})$
is equal to $t$.

It is well known that linear codes have important applications in consumer electronics, communication
and data storage system. Besides, linear codes with  few weights have also applications in secret sharing \cite{CDY05,Yuan-Ding}, authentication codes \cite{Dingau}, association schemes \cite{Calderbank}, and strongly regular graphs \cite{Calderbank}.
Very recently, Ding \textit{et al.} proposed a general construction of linear codes from a subset $D$ of $\gf_{p^m}$ and
the trace function from $\gf_{p^m}$ to $\gf_p$ \cite{Ding1,Ding2}. This construction can generate two-weight and three-weight linear codes
with excellent parameters if the subset $D$ is appropriately chosen.

The objective of this paper is to present a construction of two-weight or three-weight linear codes based on quadratic Bent functions.
It works for any quadratic Bent function over $\gf_p$, and includes the construction in \cite{Ding2} as a special case.
The weight distribution of the resultant linear codes are determined.
Some of the linear codes obtained in this paper are optimal in the sense that they meet some bounds
on linear codes.

The rest of this paper is organized as follows.  Section \ref{sec-pre}
introduces basic theory of quadratic forms over finite fields which will be
needed in subsequent sections.  Section \ref{sec-main}
establishes a bridge from quadratic Bent functions to linear codes with two or three weights, and
settles the weight distributions of linear codes from quadratic Bent functions. Finally, Section \ref{sec-con}
concludes this paper and makes some comments.

\section{Quadratic forms over finite fields}\label{sec-pre}
Identifying $\gf_{p^m}$ with the $m$-dimensional
$\gf_p$-vector space $\gf_p^m$, a function $Q(x)$ from $\gf_{p^m}$ to $\gf_p$ can be regarded as an $m$-variable polynomial over $\gf_p$.
The former is called a quadratic form over $\gf_p$ if the latter is a homogeneous polynomial of degree two in the form
\begin{eqnarray*}\label{eqn-Q-homogenous}
Q(x_1,x_2,\cdots,x_m)=\sum_{1\leq i\leq j\leq m}a_{ij}x_ix_j,
\end{eqnarray*}
where $a_{ij}\in \gf_p$, and we use a basis $\{\beta_1,\beta_2,\cdots,\beta_{m}\}$ of $\gf_{p^m}$ over $\gf_p$ and identify ${x}=\sum_{i=1}^mx_i\beta_i$ with the vector $\bar{x}=(x_1,x_2,\cdots,x_{m})\in \gf_p^m$. We write $\bar{x}$ when an element is to be thought of as a vector
in $\gf_p^m$, and write $x$ when the same vector is to be thought of as an element of $\gf_{p^m}$.
The rank of the
quadratic form $Q(x)$ is defined as the codimension of the $\gf_p$ -vector space
\begin{eqnarray*}
V=\{y\in \gf_{p^m}: ~Q(x+y)-Q(x)-Q(y)=0 \textrm{~for~all~}x\in \gf_{p^m}\}.
\end{eqnarray*}
That is $|V|=p^{m-r}$ where $r$ is the rank of $Q(x)$.

Quadratic forms  have been well studied (see \cite{Niddle}, \cite{Klappereven}, \cite{Klapperodd}, for example). Here we follow the treatment in \cite{Klappereven} and \cite{Klapperodd}.
It should be noted that the rank of a quadratic form over $\gf_p$ is the smallest number of variables required to represent the quadratic form, up
to  nonsingular coordinate transformations.
Mathematically, any quadratic
form of rank $r$ can be transferred to  three canonical forms as follows. Throughout this section, let $B_{2j}(\bar{x})=x_1x_2+x_3x_4+\cdots+x_{2j-1}x_{2j}$ where $j\geq 0$ is an integer (we assume that $B_0=0$ when $j=0$).
Let $\nu(x)$ be a function over $\gf_p$ defined by $\nu(0)=p-1$  and $\nu(\zeta)=-1$ for any
$\zeta\in \gf_p^*$.

\begin{lemma}\label{Lemma_Klapper2}(\cite{Klapperodd})
Let $Q(x)$ be a quadratic form over $\gf_p$ of rank $r$ in $m$ variables. Then $Q(x)$ is equivalent (under a
change of coordinates) to
one of the following three standard types:

\vspace{2mm}
~~~~Type I: ~~~~ $B_r(\bar{x})$,~~~$r$~~even;

\vspace{2mm}

~~~~Type II: ~~~ $B_{r-1}(\bar{x})+\mu x_m^2$,~~~$r$~~odd;
\vspace{2mm}

~~~~Type III: ~~ $B_{r-2}({\bar x})+x^2_{r-1}-\varsigma x_r^2$,~~~$r$~~even;\\
where $\mu\in \{1, \varsigma\}$ and $\varsigma$ is a fixed nonsquare in $\gf_{p}$.
Furthermore, for any $\zeta\in \gf_p$, the number of solutions $\bar{x}\in \gf_{p}^m$ to the equation $Q(\bar{x})=\zeta$ is:

~~~~Type I: ~~~~ $p^{m-1}+\nu(\zeta)p^{m-r/2-1}$;
\vspace{2mm}

~~~~Type II: ~~~ $p^{m-1}+\eta(\mu \zeta)p^{m-(r+1)/2}$;

\vspace{2mm}
~~~~Type III: ~~ $p^{m-1}-\nu(\zeta)p^{m-r/2-1}$;\\
where $\eta$ is the quadratic (multiplicative) character of $\gf_p$ and $\eta(0)$ is assumed to be 0.
\end{lemma}

An interesting class of quadratic forms is the quadratic form with full rank since in this case the corresponding functions are Bent functions.
 Let $f$ be a function from $\gf_{p^m}$ to $\gf_{p}$. The Walsh transform of $f$ at the point $\lambda\in\gf_{p^m}$ is defined as
   \begin{eqnarray*}
    \widehat{f}(\lambda)=\sum\limits_{x \in \gf_{p^m}}\omega_p^{f(x)-\tr_1^m(\lambda  x)},
  \end{eqnarray*}
  where $\omega_p=e^{2\pi\sqrt{-1}/p}$ is a primitive $p$-th root of unity and $\tr_1^m(x)=\sum^{m-1}_{i=0}x^{p^{i}}$ is the trace function from $\gf_{p^m}$ to $\gf_{p}$.

The function $f$ is called a {Bent function} if $|\widehat{f}(\lambda)|=p^{m/2}$ for all $\lambda\in\gf_{p^m}$. Bent function was introduced by Rothaus in \cite{Rothaus} for boolean functions, namely the case of $p=2$, and later was generalized by Kumar, Scholtz, and Welch in \cite{Kumar-SW} for $p>2$.

%Let $f:~\gf_{p^m}\rightarrow \gf_p$ be a $p$-ary function in $m$ variables. It is said to be a Bent function if
%$|\widehat{f}(\lambda)|=p^{m/2}$ for all $\lambda\in \gf_{p^m}$,
%where
%  \begin{eqnarray*}
%    \widehat{f}(\lambda)=\sum\limits_{x \in {\gf}_{p^m}}{\omega_p}^{\tr^m_1(f(x)-\lambda x)}
%  \end{eqnarray*}
%is called the Walsh spectrum of $f(x)$ at the point $\lambda$.
%
The following result was proven in \cite{Tang}.

\begin{lemma}(\cite{Tang})\label{lemma-fullrank} %,\cite{Tang}
Let $Q(x)$ be a quadratic form from $\gf_{p^m}$ to $\gf_{p}$ with full rank $m$. Then
$$
\left| \sum_{x\in \gf_{p^m}}\omega_p^{Q(x)-\tr^m_1\left(\lambda x\right)}\right|=p^{m/2}
$$
for any $\lambda \in \gf_{p^m}$.
\end{lemma}

It can be readily verified from Lemma \ref{lemma-fullrank} that a quadratic form  $Q(x)$ from $\gf_{p^m}$ to $\gf_{p}$ is a Bent function if and only if it has full rank.
In the next section, we will employ quadratic Bent functions to construct linear codes with few weights.  Before doing this, we first give
two lemmas that will be used to prove the main result of the paper.

The following  follows directly from Lemma \ref{Lemma_Klapper2}.
\begin{lemma}\label{lemma-mainlemma_1}
Let  $Q(x)$  be a quadratic Bent function from $\gf_{p^m}$ to $\gf_p$. Define
$$
D_Q=\{x\in \gf_{p^m}^*: Q(x)=0\}.
$$
Then
$$
|D_Q|=p^{m-1}-1
$$
if $m$ is odd, and otherwise
\begin{eqnarray}\label{eqn-size-DQ}
|D_Q|=p^{m-1}+\epsilon(p-1)p^{\frac{m-2}{2}}-1,
\end{eqnarray}
here and hereinafter $\epsilon=1$ if $Q(x)$ is equivalent to Type I and $\epsilon=-1$ if $Q(x)$ is equivalent to Type III.
\end{lemma}

\begin{lemma}\label{Lemma-mainlemma-2}
Let  $Q(x)$  be a quadratic Bent function from $\gf_{p^m}$ to $\gf_p$. For any $b\in \gf_{p^m}$, define
$$\label{eqn_def_Nb}
D_{Q,b}=\{x\in \gf_{p^m}^*: Q(x)=0~ \mbox{and}~ \tr_1^m(bx)=0\}
$$
and
$$
N_b=|D_{Q,b}|.
$$
Then $N_b$ has the following distribution as $b$ runs through $\gf_{p^m}$:
\begin{eqnarray*}
N_b=\left\{
\begin{array}{ll}
p^{m-1}+\epsilon(p-1)p^{\frac{m-2}{2}}-1, &1 \textrm{~time}\\
p^{m-2}-1, &(p-1)\left(p^{m-1}-\epsilon p^{\frac{m-2}{2}}\right)\textrm{~times} \\
p^{m-2}+\epsilon (p-1)p^{\frac{m-2}{2}}-1, &p^{m-1}+\epsilon(p-1)p^{\frac{m-2}{2}}-1\textrm{~times}
\end{array} \right.\ \
\end{eqnarray*}
if $m$ is even, and otherwise
\begin{eqnarray*}
N_b=\left\{
\begin{array}{ll}
p^{m-1}-1, &1 \textrm{~time}\\
p^{m-2}-1, &p^{m-1}-1\textrm{~times} \\
p^{m-2}+(p-1)p^{\frac{m-3}{2}}-1, &\frac{p-1}{2}\left(p^{m-1}+p^{\frac{m-1}{2}}\right)\textrm{~times}\\
p^{m-2}-(p-1)p^{\frac{m-3}{2}}-1, &\frac{p-1}{2}\left(p^{m-1}-p^{\frac{m-1}{2}}\right)\textrm{~times} .
\end{array} \right.\ \
\end{eqnarray*}
\end{lemma}

\begin{proof}
When $b=0$, it is clear that
\begin{eqnarray*}\label{eqn-N00-0}
N_b=N_0=|D_Q|.
\end{eqnarray*}
The value of $N_0$ is thus determined due to Lemma \ref{lemma-mainlemma_1}.
Therefore we only need to calculate $N_b$ for $b\in \gf_{p^m}^*$. To this end, we suppose that $\{\alpha_1,\alpha_2,\cdots,\alpha_{m}\}$ and  $\{\beta_1,\beta_2,\cdots,\beta_{m}\}$ are dual basis
of $\gf_{p^m}$ over $\gf_p$. Using these bases, we write
$x=x_1\beta_1+x_2\beta_2+\cdots+x_m\beta_m$ and $b=b_1\alpha_1+b_2\alpha_2+\cdots+b_m\alpha_m$ for $x,b\in \gf_{p^m}$, where
$\bar{x}=(x_1,x_2,\cdots,x_m)\in \gf_p^m$ and $\bar{b}=(b_1,b_2,\cdots,b_m)\in \gf_p^m$.
Then we have
\begin{eqnarray}\label{eqn-N00-00}
N_b=N(0,0)-1,
\end{eqnarray}
where $N(0,0)$ is the number of solutions $\bar{x}\in \gf_p^m$ to the  equation system
\begin{eqnarray*}\label{eqn-N00-1}
\left\{
\begin{array}{ll}
Q(\bar{x})=0\\
\bar{b}\cdot \bar{x}=0
\end{array} \right.\ \
\end{eqnarray*}
where $\bar{b}\cdot \bar{x}=b_1x_1+b_2x_2+\cdots+b_mx_m$ is the inner product of the vectors $\bar{b}$ and $\bar{x}$.
Let
\begin{eqnarray*}\label{eqn-N00-1}
\hat{Q}(\bar{x})=\left\{
\begin{array}{ll}
B_{m}(\bar{x}),& \mbox{~if~}Q(x)\mbox{~is~equivalent~ to~Type~I}  \\
B_{m-1}(\bar{x})+\frac{ x_m^2}{4\mu} , &\mbox{~if~}Q(x)\mbox{~is~equivalent~ to~Type~II}\\
B_{m-2}(\bar{x})+\frac{x_{m-1}^2}{4}-\frac{x_{m}^2}{4\varsigma} , &\mbox{~if~}Q(x)\mbox{~is~equivalent~ to~Type~III}
\end{array} \right.\ \
\end{eqnarray*}
where  $\mu\in \{1, \varsigma\}$ and $\varsigma$ is a fixed nonsquare in $\gf_{p}$, as defined in Lemma \ref{Lemma_Klapper2}. Note that $\hat{Q}(\bar{x})$ is equivalent to $Q(\bar{x})$ under a change of coordinates.
Thus $\hat{Q}(\bar{x})$ and $Q(\bar{x})$ are equivalent to the same standard type.
Thanks to Proposition 3.4 in \cite{Klapperodd}, we have
\begin{eqnarray}\label{eqn-N00-2}
N(0,0)=\left\{
\begin{array}{ll}
p^{m-2}+\epsilon(p-1)p^{\frac{m-2}{2}}, &\mbox{~if~} \hat{Q}(\bar{b})=0\\
p^{m-2}, &\mbox{~if~} \hat{Q}(\bar{b})\neq 0
\end{array} \right.\ \
\end{eqnarray}
if $m$ is even and otherwise
\begin{eqnarray}\label{eqn-N00-3}
N(0,0)=\left\{
\begin{array}{ll}
p^{m-2}, &\mbox{~if~}  \hat{Q}(\bar{b})=0\\
p^{m-2}+\eta(\mu \hat{Q}(\bar{b})) (p-1)p^{\frac{m-3}{2}}, &\mbox{~if~}  \hat{Q}(\bar{b})\neq 0
\end{array} \right.\ \
\end{eqnarray}
where $\eta$ is the quadratic  character of $\gf_p$ and $\eta(0)$ is assumed to be 0. By (\ref{eqn-N00-00}), the value distribution of $N_b$ for even $m$ (resp., odd $m$) then follows from Equation (\ref{eqn-N00-2}) (resp., (\ref{eqn-N00-3})),
and the number of solutions $\bar{b}\in \gf_{p}^m$ to $\hat{Q}(\bar{b})=\zeta$ given in Lemma \ref{Lemma_Klapper2}, where $\zeta\in \gf_p$.
\end{proof}

\section{Linear Codes with Two or Three Weights From Quadratic Bent Functions}\label{sec-main}
In this section, inspired by the work of Ding \emph{et al.} \cite{Ding1,Ding2}, we shall construct
several classes of linear codes with two or three weights employing
quadratic forms over finite field $\gf_p$. Before doing this, we give a brief
introduction to the construction of linear codes proposed by Ding \textit{et al.} recently \cite{Ding1}, \cite{Ding2}.

Let $D=\{d_0,d_1,\cdots,d_{n-1}\}$ be any $n$-subset of $\gf_{p^m}$. Define a linear code $\C_D$ of length $n$ from $D$ as follows:
\begin{eqnarray}\label{eqn-defn-CD}
\C_D:=\{\bc_b: b\in \gf_{p^m}\},
\end{eqnarray}
where
\begin{eqnarray}\label{eqn-def-trace}
{\bf c}_b=(\tr_1^m(bd_0), \tr_1^m(bd_1), \cdots, \tr_1^m(bd_{n-1})).
\end{eqnarray}

Clearly, the dimension of $\C_D$ is at most $m$. In general, it is difficult to determine the
minimal distance of $\C_D$ not to mention the weight distribution. However,
the weight distribution of $\C_D$ can be settled in some cases \cite{Ding1}, \cite{Ding2}.
For example, when $D=\{x\in \gf_{p^m}^*: \tr_1^m(x^2)=0\}$ and $p$ is an odd prime,
the weight distribution of $\C_D$ was completely determined in \cite{Ding2}.
It turns out in \cite{Ding2} that $\C_D$ is two-weight for even $m$ and three-weight for odd $m$.
Note that $\tr_1^m(x^2)$ is a quadratic Bent function over $\gf_p$. This inspires us
to construct linear code  from general quadratic Bent functions  over $\gf_p$.

Let $Q(x)$ be a quadratic Bent function from $\gf_{p^m}$ to $\gf_p$. Define
\begin{eqnarray}\label{eqn-D-Q}
D_Q=\{x\in \gf_{p^m}^*: Q(x)=0\},
\end{eqnarray}
and a linear code $\C_{D_Q}$ according to (\ref{eqn-defn-CD}).
For the code $\C_{D_Q}$, we have the following results.

\begin{table}[ht]
\parbox{6.5 cm}{\caption{The weight distribution of $\C_{D_Q}$ for odd $m$.}}\label{tab-CG2}
\centering
\begin{tabular}{|l|l|}
\hline
Weight $w$    & No. of codewords $A_w$  \\ \hline
$0$                                                        & $1$ \\ \hline
$(p-1)(p^{m-2}-p^{\frac{m-3}{2}})$          & $\frac{p-1}{2}(p^{m-1}+p^{\frac{m-1}{2}})$ \\ \hline
$(p-1)p^{m-2}$                               & $p^{m-1}-1$ \\ \hline
$(p-1)(p^{m-2}+p^{\frac{m-3}{2}})$          &  $\frac{p-1}{2}(p^{m-1}-p^{\frac{m-1}{2}})$  \\ \hline
\end{tabular}
\end{table}

\begin{table}[ht]
\caption{The weight distribution of $\C_{D_Q}$ for even $m$.}\label{tab-CG2}
\centering
\begin{tabular}{|l|l|}
\hline
Weight $w$    & No. of codewords $A_w$  \\ \hline
$0$                                                        & $1$ \\ \hline
$(p-1)p^{m-2}$                               & $p^{m-1}+\epsilon(p-1)p^{\frac{m-2}{2}}-1$  \\ \hline
$(p-1)(p^{m-2}+\epsilon p^{\frac{m-2}{2}})$          & $(p-1)(p^{m-1}-\epsilon p^{\frac{m-2}{2}})$  \\ \hline
\end{tabular}
\end{table}

\begin{theorem}\label{Th-Main1}
If $m$ is odd,  then
$\C_{D_Q}$ is a three-weight $[p^{m-1}-1, m]$  code over $\gf_p$ with the weight distribution in Table 1.
\end{theorem}
\begin{proof}
According to the definition of $\C_{D_Q}$, its length is equal to $|D_Q|$.
By Lemma \ref{lemma-mainlemma_1},  $|D_Q|=p^{m-1}-1$ when $m$ is odd.
For any codeword $\bc_b$ in $\C_{D_Q}$, according to the definition, its Hamming weight is equal to
\begin{eqnarray*}
{\textrm{WT}}({\bf c}_{b})&=&|D_Q|-|D_{Q,b}|
\end{eqnarray*}
where
$$
D_{Q,b}=\{x\in \gf_{p^m}^*: Q(x)=0~ \mbox{and}~ \tr_1^m(bx)=0\}.
$$
Then, the weight distribution of $\C_{D_Q}$ follows from Lemmas \ref{lemma-mainlemma_1}
and  \ref{Lemma-mainlemma-2}. Finally, the dimension of $\C_{D_Q}$
follows from its weight distribution.
\end{proof}

\begin{theorem}\label{Th-Main2}
If $m$ is even,  then
$\C_{D_Q}$ is a two-weight $[p^{m-1}+\epsilon(p-1)p^{\frac{m-2}{2}}-1, m]$  code over $\gf_p$ with the weight distribution in Table 2,
where $\epsilon=1$ if $Q(x)$ is equivalent to Type I and $\epsilon=-1$ if $Q(x)$ is equivalent to Type III.
\end{theorem}
\begin{proof}
The proof of this theorem is similar to that of Theorem \ref{Th-Main1}.
\end{proof}

Theorems \ref{Th-Main1} and \ref{Th-Main2} imply that any quadratic Bent function over $\gf_p$ naturally gives a two-weight or three-weight linear code. In the remainder of this section, we shall introduce several classes of linear codes from some known quadratic Bent functions.

\subsection{Linear Codes From Some Known Planar Functions}
A  function $\pi(x)$ from $\gf_{p^m}$ to $\gf_{p^m}$ is referred to as perfect
nonlinear  if
$$
\max_{a \in \gf_{p^m}^*} \max_{b \in \gf_{p^m}} |\{x \in \gf_{p^m}: \pi(x+a)-\pi(x)=b\}| =1.
$$
A perfect nonlinear function from a finite field to itself is also called a planar function
in finite geometry \cite{Coulter-Matthews}. Some known quadratic planar functions from $\gf_{p^m}$ to $\gf_{p^m}$ are summarized as follows
\begin{itemize}
\item [(a)] $\pi(x)=x^2$;
\item [(b)] $\pi(x)=x^{p^k+1}$ where $m/\gcd(m,k)$ is odd \cite{Dembowski-Ostrom};
\item [(c)] $\pi(x)=x^{10}-x^6-x^2$ where $p=3$ and $m$ is odd \cite{Coulter-Matthews};
\item [(d)]$\pi(x)=x^{10}-ux^6-u^2x^2$ where $p=3$, $m$ is odd and $u\in\gf_{p^m}^*$ \cite{Ding-Yuan};
\item [(e)]$\pi(x)=x^{p^s+1}-u^{p^k-1}x^{p^k+p^{2k+s}}$ where $m=3k$,  $\gcd(k,3)=1$, $k-s \equiv 0~(\bmod 3)$, $s\neq k$ and $k/\gcd(k,s)$ is odd, and $u$ is a primitive element of $\gf_{p^{m}}$ \cite{Zha}.
\end{itemize}

It is well known that every component function $\tr_1^m(c\pi(x)), c\in\gf_{p^m}^*$ of a planar function $\pi(x)$ over $\gf_{p^{m}}$ is a Bent function \cite{CD04}. Thus, for any planar function $\pi(x)$ listed as above, one obtains that $Q(x)=\tr_1^m(c\pi(x))$ is a quadratic Bent function over $\gf_p$. Using
these planar functions, we can obtain linear codes with two or three weights
according to Theorems \ref{Th-Main1} and \ref{Th-Main2}.

\begin{corollary}\label{coro-planar}
Let $\pi(x)$ be any planar function listed above and $Q(x)=\tr_1^m(c\pi(x))$, where $c\in\gf_{p^m}^*$. Then
\begin{enumerate}
  \item  $\C_{D_Q}$ is a three-weight $[p^{m-1}-1, m]$  code over $\gf_p$ with the weight distribution in Table 1 if $m$ is odd; and
  \item $\C_{D_Q}$ is a  two-weight $[p^{m-1}+\epsilon (p-1)p^{\frac{m-2}{2}}-1, m]$  code over $\gf_p$ with the weight distribution in Table 2  if $m$ is even. Furthermore, $\epsilon=\eta(c)(-1)^{(\frac{p-1}{2})^2\frac{m}{2}+1}$ for the  planar functions listed in (a) and (b).
\end{enumerate}
\end{corollary}

\begin{proof}
According to Theorem \ref{Th-Main2}, we only need to prove $\epsilon=\eta(c)(-1)^{(\frac{p-1}{2})^2\frac{m}{2}+1}$ for the planar functions listed in (a) and (b).
When $\pi(x)=x^2$, similar as the proof of Theorem 2 in \cite{Ding2}, one can easily obtain $\epsilon=\eta(c)(-1)^{(\frac{p-1}{2})^2\frac{m}{2}+1}$
for $Q(x)=\tr_1^m(cx^2)$.
When $\pi(x)=x^{p^{k}+1}$  where $m/\gcd(m,k)$ is odd. Note that $\gcd(p^m-1, p^k+1)=2$. We have
$$
|\{x\in \gf_{p^m}^*: \tr_1^m(cx^{p^{k}+1})=0\}|=|\{x\in \gf_{p^m}^*: \tr_1^m(cx^{2})=0\}|.
$$
By Lemma \ref{lemma-mainlemma_1},  $\epsilon=\eta(c)(-1)^{(\frac{p-1}{2})^2\frac{m}{2}+1}$ for $Q(x)=\tr_1^m(cx^{p^{k}+1})$.
\end{proof}

%\begin{open}
%Determine the sign of $\epsilon$ for  the planar function given in (e).
%\end{open}

It will be nice if the sign of $\epsilon$ for the planar function given in (e) with even $m$  can be determined.
This can be done if we can determine the equivalent type of the corresponding Bent function.

\begin{example}\label{example-planar-odd}
Let $p=3$, $m=5$, and $Q(x)=\tr_1^m(x^{10}-x^6-x^2)$. The Magma program shows that
$\C_{D_Q}$ has parameters $[80,5,48]$ and weight enumerator $1+90z^{48}+80z^{54}+72z^{60}$, which
agrees with the result in Corollary \ref{coro-planar}.
\end{example}

\begin{example}\label{example-planar-even}
Let $p=3$, $m=6$, $\beta$ be a primitive element of $\gf_{3^6}$. When $Q(x)=\tr_1^m(x^{p^2+1})$, the Magma program shows that
$\C_{D_Q}$ has parameters $[224,6,144]$ and weight enumerator $1+504z^{144}+224z^{162}$.
When  $Q(x)=\tr_1^m(\beta x^{p^2+1})$, the Magma program shows that
$\C_{D_Q}$ has parameters $[260,6,162]$ and weight enumerator $1+260^{162}+468z^{180}$.
The computer experimental data agrees with the result in Corollary \ref{coro-planar}.
\end{example}

\subsection{Linear Codes From Gold Class of Bent Functions}

 Let $p$ be an odd prime and $c=\alpha^t\in\gf_{p^m}^*$, where $\alpha$ is a primitive element of $\gf_{p^m}$ and $t$ is an integer with $0\leq t\leq p^m-2$. Then for any $j\in\{1,2,\cdots,m\}$, Helleseth and Kholosha in \cite{Helleseth-K2006} proved that the quadratic function
 \begin{eqnarray}\label{eq-Qgold}
  Q(x)=\tr_1^{m}(cx^{p^{j}+1})
 \end{eqnarray}
  is a Bent function if and only if
  \begin{eqnarray}\label{eq-gold-condi}
  p^{\gcd(2j,m)}-1\nmid\frac{p^m-1}{2}-t(p^j-1).
 \end{eqnarray}

 \begin{corollary} \label{cor-Goldcase}
 Let $Q(x)$ be defined as \eqref{eq-Qgold} and it satisfies  \eqref{eq-gold-condi}. Then
$\C_{D_Q}$ is a three-weight $[p^{m-1}-1, m]$ code over $\gf_p$ with the weight distribution in Table 1 if $m$ is odd, and for even $m$, $\C_{D_Q}$ is a  two-weight $[p^{m-1}+\epsilon (p-1)p^{\frac{m-2}{2}}-1, m]$  code over $\gf_p$ with the weight distribution in Table 2.
\end{corollary}

Observe that the Gold class of quadratic Bent functions defined by \eqref{eq-Qgold} covers several known cases:
 \begin{enumerate}
   \item Sidelnikov Bent function: when $j=m$, then $Q(x)$ is reduced to $Q(x)=\tr_1^{m}(cx^2)$;
   \item Kumar-Moreno Bent function: Kumar and Moreno in \cite{Kumar-Moreno} showed that $f(x)=\tr_1^{m}(x^{p^{k}+1})$ is a Bent function, where $m/\gcd(m,k)$ is odd and $c\in\gf_{p^m}^*$.
   \item Kasami Bent function: when $j=m/2$, then $Q(x)$ is reduced to $Q(x)=\tr_1^{m}(cx^{p^{m/2}+1})$ which is a Bent function if $c+c^{p^{m/2}}\ne 0$ \cite{Liu}.
 \end{enumerate}

\begin{remark}
The  Sidelnikov Bent function and the Kumar-Moreno Bent function are exactly the Bent functions
from the planar functions $\pi(x)=x^2$ and $\pi(x)=x^{p^k+1}$ mentioned in above subsection.
\end{remark}

When $m$ is even, one should also note that the sign of $\epsilon$ can be determined by the value of the Walsh transform of $Q(x)$ at the zero point. Let
$$
\mathbb{N}_i=|\{x\in\gf_{p^m}:Q(x)=0\}|
$$
for $i=0,1,\cdots,p-1$, then
$$
\widehat{Q}(0)=\sum_{x\in\gf_{p^m}}\omega_p^{Q(x)}=\sum_{i=0}^{p-1}\mathbb{N}_i\omega_p^{i}.
$$ Thus, the values of $\mathbb{N}_i$ for $i=0,1,\cdots, p-1$ can be determined by the value of $\widehat{Q}(0)$ and the well known fact that the polynomial $1+x+x^2+\cdots+x^{p-1}$ is irreducible over the rational number field. Therefore, the sign of $\epsilon$ can be determined by comparing the values of $\mathbb{N}_0$ and $|D_Q|$ given as in \eqref{eqn-size-DQ}. This fact implies that the sign of $\epsilon$ in Corollary \ref{cor-Goldcase} can be determined based on Lemma 2 given in \cite{Helleseth-K2006} for any given parameters $p,n,j$ and $c$. Using this method,  the sign of $\epsilon$ for the Kasami Bent function can be directly determined as follows.
 \begin{corollary}\label{coro-kasa}
Let $m$ be even and $Q(x)=\tr_1^{m}(cx^{p^{m/2}+1})$ with $c+c^{p^{m/2}}\ne 0$.
Then $\C_{D_Q}$ is a two-weight $[p^{m-1}+\epsilon (p-1)p^{\frac{m-2}{2}}-1, m]$ code over $\gf_p$ with the weight distribution in Table 2 where $\epsilon=-1$.
\end{corollary}

\begin{proof}
  According to Theorem \ref{Th-Main2}, it is sufficient to show that $\epsilon=-1$ for the Kasami Bent function.
  Note that $x^{p^{m/2}+1}$ runs through each element
of $\gf_{p^{m/2}}^*$ exactly $p^{m/2+1}$ times as $x$ ranges over $\gf_{p^m}^*$.
Thus for each $y\in \gf_p^*$, we have
\begin{eqnarray*}
\sum_{x\in \gf_{p^m}^*}\omega_p^{\tr_1^{m}(ycx^{p^{m/2}+1})}=(p^{m/2}+1)\sum_{z\in \gf_{p^{m/2}}^*}\omega_p^{\tr_1^{m}(ycz)}
=-1-p^{m/2}.
\end{eqnarray*}
It then follows that
\begin{eqnarray*}\label{eqn-kasami-1}
|\{x\in \gf_{p^m}^*: Q(x)=0)\}|&=&{1\over p}\sum_{x\in \gf_{p^m}^*}\sum_{y\in \gf_p}\omega_p^{y\tr_1^m(cx^{p^{m/2}+1})}\nonumber\\
&=&\frac{1}{p}\left(p^m-1+\sum_{y\in \gf_p^*}\sum_{x\in \gf_{p^m}^*}\omega_p^{\tr_1^m(ycx^{p^{m/2}+1})}\right)\nonumber\\
&=&\frac{1}{p}(p^m-1-(p-1)(p^{m/2}+1))\nonumber\\
&=&p^{m-1}-(p-1)p^{\frac{m-2}{2}}-1.
\end{eqnarray*}
Comparing this value with (\ref{eqn-size-DQ}), one obtains that $\epsilon=-1$. This completes the proof.
\end{proof}

\begin{example}\label{example-kasa-even}
Let $p=3$, $m=4$ and $Q(x)=\tr_1^m(x^{p^2+1})$. The Magma program shows that
$\C_{D_Q}$ has parameters $[20,4,12]$ and weight enumerator $1+60z^{12}+20z^{18}$, which agrees with
the result in Corollary \ref{coro-kasa}. This code is optimal due to the Griesmer bound.
\end{example}

\begin{example}\label{example-kasa-even-2}
Let $p=5$, $m=4$ and $Q(x)=\tr_1^m(x^{p^2+1})$. The Magma program shows that
$\C_{D_Q}$ has parameters $[104,4,80]$ and weight enumerator $1+520z^{80}+104z^{100}$, which agrees with
the result in Corollary \ref{coro-kasa}. This code is almost optimal since the best linear code of length 104 and dimension
4 over $\gf_5$ has minimal weight 81.
\end{example}

\subsection{Linear Codes From the Helleseth-Gong Function}

The Helleseth-Gong (HG) function $H(x)$
from $\gf_{p^m}$ to $\gf_p$ is defined  by \cite{Helleseth}
\begin{eqnarray}\label{eqn-hg-function}
H(x)=\tr^m_1 \left(\sum_{i=0}^{\ell}u_{i}x^{(p^{2 i}+1)/2}\right)
\end{eqnarray}
where  $m=2\ell+1$, $1\le s\leq 2\ell$ is an integer such that
$\gcd(s,2\ell+1)=1$, $b_{0}=1$, $b_{is}=(-1)^{i}$ and
$b_{i}=b_{2\ell+1-i}$ for $i=1,2,\cdots,\ell$, $u_{0}=b_{0}/2=(p+1)/2$,
and $u_{i}=b_{2i}$ for $i=1,2,\cdots,\ell$. Herein, all the indexes of
$b$'s are taken mod $(2\ell+1)$. The following result
was proved by Jang {\em et al.} (\cite{Jang}, p. 1842).

\begin{lemma}\label{lem-hg-function}
Let $H(x)$ be the HG function defined by (\ref{eqn-hg-function}). Then
$Q(x)=H(x^2)$ is a quadratic Bent function.
\end{lemma}

The following follows immediately from Theorem \ref{Th-Main2} and Lemma \ref{lem-hg-function}.

\begin{corollary}\label{coro-hg}
Let $m$ be odd and $Q(x)=H(x^2)$ where $H(x)$ is the HG function defined by (\ref{eqn-hg-function}). Then $\C_{D_Q}$ is a three-weight $[p^{m-1}-1, m]$  code over $\gf_p$ with the weight distribution in Table 1.
\end{corollary}

\begin{example}\label{example-hg-odd}
Let $p=3$, $m=5$ and and the HG function in (\ref{eqn-hg-function}) be given by
$H(x)=\tr_1^5(2x+2x^{5}+x^{41})$.  Then $Q(x)=\tr_1^5(2x^2+2x^{10}+x^{82})$.
The Magma program shows that
$\C_{D_Q}$ has parameters $[80,5,48]$ and weight enumerator $1+90z^{48}+80z^{54}+72z^{60}$, which
agrees with the result in Corollary \ref{coro-hg}.
\end{example}

%\subsection{Linear Codes From Bent Functions of Dillon Exponents}
%
%
%\begin{corollary}
%Let $m$ be even and $\ell=\frac{m}{2}$. Let $Q(x)=\tr_1^{\ell}(x^{p^\ell+1})$.
%Then $\C_{D_Q}$ is a two-weight $[p^{m-1}+\epsilon(p-1)p^{\frac{m-2}{2}}-1, m]$  code over $\gf_p$ with the weight distribution in Table 2
%where $\epsilon=-1$.
%\end{corollary}
%
%
%\subsection{Linear Codes From Bent Functions of Dillon Exponents}
%
%{\color{red} please write this part}
%

\subsection{Linear Codes From Quadratic Bent Function in Polynomial Form}

In general, up to equivalence (Section IV, \cite{Helleseth-K2006}), any quadratic function having no linear term over $\gf_{p^m}$ can be expressed as the form of
\begin{eqnarray}\label{eq-quadr-poly}
  Q(x)=\sum_{i=0}^{\lfloor m/2\rfloor}\tr_1^m(c_ix^{p^i+1}),
\end{eqnarray}
where $\lfloor x\rfloor$ denotes the largest integer not exceeding $x$ and $c_i\in\gf_{p^m}$ for $i=0,1,\cdots,\lfloor m/2\rfloor$.

For an odd prime $p$, Helleseth and Kholosha proved that $Q(x)$ defined by \eqref{eq-quadr-poly} is Bent if and only if a corresponding $m\times m$
symmetric matrix is nonsingular \cite{Helleseth-K2006}. Normally, it is difficult to determine whether a matrix of order $m$ has full rank or not. But for some special cases, for example, the case of $c_i\in\gf_{p}$ for $i=0,1,\cdots,\lfloor m/2\rfloor$, the Bentness of $Q(x)$ defined by \eqref{eq-quadr-poly} can be determined easier \cite{Helleseth-K2006,Khoo-GS}. Following the line of this work, Li, Tang and Helleseth presented a large number of Bent functions of the form \eqref{eq-quadr-poly} with $c_i\in\gf_{p}$ for $i=0,1,\cdots,\lfloor m/2\rfloor$ in a simple way \cite{LiNian}. Then, according to Theorems \ref{Th-Main1} and \ref{Th-Main2}, linear codes with two or three weights can be obtained.

\begin{corollary}\label{cor-quadr-poly}
   Let $Q(x)$ be defined as \eqref{eq-quadr-poly}. If $Q(x)$ is Bent, then
$\C_{D_Q}$ is a three-weight $[p^{m-1}-1, m]$ code over $\gf_p$ with the weight distribution in Table 1 if $m$ is odd, and for even $m$, $\C_{D_Q}$ is a  two-weight $[p^{m-1}+\epsilon (p-1)p^{\frac{m-2}{2}}-1, m]$  code over $\gf_p$ with the weight distribution in Table 2.
\end{corollary}

\begin{example}\label{example-polynomial-odd}
Let $p=3$, $m=5$ and $Q(x)=\tr_1^m(x^2+2x^{p+1}+x^{p^2+1})$. According to Corollary 11 in \cite{LiNian}, $Q(x)$ is a Bent function in $\gf_{3^5}$.
The Magma program shows that
$\C_{D_Q}$ has parameters $[80,5,48]$ and weight enumerator $1+90z^{48}+80z^{54}+72z^{60}$, which
agrees with the result in Corollary \ref{cor-quadr-poly}.
\end{example}

\begin{remark}\label{remark-sign}
Notice that Proposition 1 in \cite{Helleseth-K2006} gave an explicit expression for the Walsh transform values of $Q(x)$ defined by \eqref{eq-quadr-poly} based on the dual of $Q(x)$ and the determinant of $Q(x)$ (i.e., the determinant of the corresponding matrix associated with $Q(x)$). However, it does not help us to determine the sign of $\epsilon$ for even $m$. This is because that one can determine which Type of $Q(x)$ is equivalent to according to Lemma \ref{Lemma_Klapper2} if one knows the determinant of $Q(x)$. Thus, the determination of the sign of $\epsilon$ in Corollary \ref{cor-quadr-poly} remains open.
The reader is invited to join the adventure.
\end{remark}

\vspace{3mm}

Finally, we conclude this section by mentioning that all the codes obtained above can be punctured into
a shorter ones whose weight distribution can be easily derived from those of the original codes.
Note that for any quadratic Bent function $Q(x)$, it is easy to verify that $Q(yx)=y^2Q(x)$ for any $y\in \gf_p$.
Thus $Q(x)=0$ means that $Q(yx)=0$ for all $y\in \gf^*_p$. Hence the set $D_Q$ of (\ref{eqn-D-Q})
can be expressed as
\begin{eqnarray}\label{eqn-D-Q-bar}
D_Q=\gf^*_p \overline{D}_Q=\{yz: y\in \gf_p^* \mbox{~~and~~} z\in \overline{D}_Q\}
\end{eqnarray}
where $z_i/z_j\notin \gf_p^*$ for each pair of distinct elements $z_i$ and $z_j$ in $\overline{D}_Q$.
This implies that $\C_{\overline{D}_Q}$ is a punctured version of $\C_{D_Q}$. Notice that for any $a\in \gf_{p^m}$,
\begin{eqnarray}\label{eqn-relation}
&&|\{x\in D_Q: Q(x)=0~\mbox{and}~\tr_1^m(ax)=0\}|\nonumber\\
&&=(p-1)|\{x\in \overline{D}_Q: Q(x)=0~\mbox{and}~\tr_1^m(ax)=0\}|.
\end{eqnarray}
We immediately have the following results for $\C_{\overline{D}_Q}$.
\begin{corollary}\label{cor-Main2-1}
Let $m$ be odd and $Q(x)$ be any quadratic Bent functions from $\gf_{p^m}$ to $\gf_p$. Then
$\C_{\overline{D}_Q}$ is a three-weight code over $\gf_p$ with parameters
$$
\left[\frac{p^{m-1}-1}{p-1}, m\right]
$$
and the weight distribution in Table 3.
\end{corollary}

\begin{corollary}\label{cor-Main2-2}
Let $m$ be even and $Q(x)$ be any quadratic Bent functions from $\gf_{p^m}$ to $\gf_p$. Then
$\C_{\overline{D}_Q}$ is a two-weight code with parameters
$$
\left[\frac{p^{m-1}-1}{p-1}+\epsilon p^{\frac{m-2}{2}}, m\right]
$$
and the weight distribution in Table 4.
\end{corollary}

\begin{remark}
The codes $\C_{\overline{D}_Q}$ in Corollaries \ref{cor-Main2-1} and \ref{cor-Main2-2} are exactly the $p$-ary projective codes from nondegenerate quadrics in
projective spaces which were studied in \cite{Wolfmann} and  \cite{Wan}.
Based on some results in projective geometry, Wan obtained the minimal weight and weight hierarchies of these linear codes (see Theorem 9 in \cite{Wan}).
To the best of our knowledge, the weight distribution of $\C_{\overline{D}_Q}$ has not been established in literature. In Corollaries \ref{cor-Main2-1} and \ref{cor-Main2-2}, employing the theory of
quadratic forms over finite fields,  we
completely determined the weight distribution of the codes $\C_{\overline{D}_Q}$. In addition, following the recent work of Ding \textit{et al.} \cite{Ding1}, \cite{Ding2}, we give the simple trace representation of the codewords in  $\C_{\overline{D}_Q}$ (see  \eqref{eqn-def-trace}) which may be
useful from the viewpoint of applications.  These are our contributions to the code $\C_{\overline{D}_Q}$.
\end{remark}

\begin{example}\label{example-kasa-even}
Let $\C_{D_Q}$ be the linear codes with parameters $[80,5,48]$ in Examples \ref{example-planar-odd}, \ref{example-hg-odd} and \ref{example-polynomial-odd}.
The Magma program shows that
$\C_{\overline{D}_Q}$ has parameters $[40,5,24]$ and weight enumerator $1+90z^{24}+80z^{27}+72z^{30}$ which agrees with the result in Corollary \ref{cor-Main2-1}. This code is optimal in the sense that
any ternary code of length 40 and dimension 5 cannot have minimal distance 25 or more \cite{Eupen}.
\end{example}

\begin{example}\label{example-kasa-even2-1}
Let $\C_{D_Q}$ be the linear codes with parameters $[20,4,12]$ in Example \ref{example-kasa-even}.
The Magma program shows that
$\C_{\overline{D}_Q}$ has parameters $[10,4,6]$ and weight enumerator $1+60z^{6}+20z^{9}$ which agrees with the result in Corollary \ref{cor-Main2-2}. This code is optimal due to the Griesmer bound.
\end{example}

\begin{example}\label{example-kasa-even2-2}
Let $\C_{D_Q}$ be the linear codes with parameters $[104,4,80]$ in Example \ref{example-kasa-even-2}.
 The Magma program shows that
$\C_{\overline{D}_Q}$ has parameters $[26,4,20]$ and weight enumerator $1+520z^{20}+104z^{25}$, which agrees with
the result in Corollary \ref{cor-Main2-2}. This code is  optimal in the sense that it meets the Griesmer bound.
\end{example}

\begin{table}[ht]
\caption{The weight distribution of $\C_{\overline{D}_Q}$ for odd $m$.}\label{tab-CG2}
\centering
\begin{tabular}{|l|l|}
\hline
Weight $w$    & No. of codewords $A_w$  \\ \hline
$0$                                                        & $1$ \\ \hline
$p^{m-2}-p^{\frac{m-3}{2}}$          & $\frac{p-1}{2}(p^{m-1}+p^{\frac{m-1}{2}})$ \\ \hline
$p^{m-2}$                               & $p^{m-1}-1$ \\ \hline
$p^{m-2}+p^{\frac{m-3}{2}}$          &  $\frac{p-1}{2}(p^{m-1}-p^{\frac{m-1}{2}})$  \\ \hline
\end{tabular}
\end{table}

\begin{table}[ht]
\caption{The weight distribution of $\C_{\overline{D}_Q}$ for even $m$.}\label{tab-CG2}
\centering
\begin{tabular}{|l|l|}
\hline
Weight $w$    & No. of codewords $A_w$  \\ \hline
$0$                                                        & $1$ \\ \hline$p^{m-2}$                               & $p^{m-1}+\epsilon(p-1)p^{\frac{m-2}{2}}-1$  \\ \hline
$p^{m-2}+\epsilon p^{\frac{m-2}{2}}$          & $(p-1)(p^{m-1}-\epsilon p^{\frac{m-2}{2}})$  \\ \hline
\end{tabular}
\end{table}

\section{Concluding Remarks}\label{sec-con}

In this paper, inspired by the work of \cite{Ding2}, quadratic Bent functions were used to construct linear codes with
few nonzero weights over finite fields. It was shown that the presented linear codes have only two or three nonzero weights if the employed quadratic Bent functions have even or odd number of variables, respectively. The weight distributions of the codes were also determined and some of constructed linear codes are optimal in the sense that their parameters meet certain bound on linear codes. The work of this paper extended the main results in \cite{Ding2}.

Notice that Lemma \ref{Lemma_Klapper2} enables us to construct linear codes along the way discussed in the paper for any quadratic function (for example, semi-bent function) over finite fields. However the minimal distance of the corresponding linear codes may not be good if the employed quadratic function is not of full rank (i.e., is not Bent). This is another motivation for us to design linear codes from quadratic Bent functions in this paper.

\section*{Acknowledgments}
The authors are very grateful to the reviewers and the
 Editor  for their comments and
suggestions that improved the presentation and quality of this
paper.  Z. Zhou's research was supported by
the Natural Science Foundation of China, Proj. No. 61201243, the Sichuan Provincial Youth Science and Technology Fund under Grant
2015JQO004, and the Open Research Fund of National Mobile Communications Research Laboratory, Southeast University under Grant 2013D10.
C. Fan's research was supported by
the Natural Science Foundation of China, Proj. No. 11571285.

%
%In this paper, quadratic Bent functions were used to construct linear codes with
%few weights over finite field, which generalized the work of \cite{Ding2}. The weight distributions of these linear codes were determined
%based on the properties of quadratic Bent functions. Some of these linear codes
%are optimal in the sense that their parameters meet certain bound on linear codes.
%
%Notice that all the quadratic functions used in the paper are of full rank.
%In general any quadratic function could be used to construct linear codes
%along the way discussed in the paper. However the minimal distance of the corresponding
%linear codes may not good if the quadratic function does not have full rank. This forms another motivation
%why we used quadratic Bent function to design linear codes.
%
%Finally we mention that there are many Bent functions  in literature other than quadratic ones.
%It would be interesting to study the linear codes from these non-quadratic Bent functions.
%The reader is invited to attack this problem.

\end{document}